\documentclass[a4paper,reqno]{amsart}
\usepackage[all]{xy}           
\usepackage{amssymb}           
\usepackage{hyperref}
\usepackage{eucal}
\usepackage[usenames,dvipsnames]{color}
\usepackage{graphicx}

\numberwithin{equation}{section}

\newtheorem{definition}{Definition}[section]
\newtheorem{lemma}[definition]{Lemma}
\newtheorem{theorem}[definition]{Theorem}
\newtheorem{proposition}[definition]{Proposition}
\newtheorem{corollary}[definition]{Corollary}
\newtheorem{remarkth}[definition]{Remark}
\newtheorem{example}[definition]{Example}
\newenvironment{remark}{\begin{remarkth}\upshape}{\hfill$\diamond$\end{remarkth}}
\renewcommand{\emph}[1]{{\bfseries\itshape{#1}}}

\usepackage{amssymb}

\newcommand{\R}{\mathbb{R}}      

\newcount\ancho \newcount\anchom \newcount\anchoa
\newcount\anchob \newcount\altura

\newcommand{\ltilde}[3][0]{\altura=0 \advance\altura by #1
           \ancho=#2 \anchom=\ancho \divide\anchom by 2
           \anchoa=\ancho \divide\anchoa by 4
           \anchob=\anchom \advance\anchob by \anchoa
           \kern-3pt \begin{array}[b]{c}
           \begin{picture}(1,1)(\anchom,-\altura)
        \qbezier(0,2)(\anchoa,5)(\anchom,2)
        \qbezier(\anchom,2)(\anchob,-1)(\ancho,4)
        \qbezier(0,2)(\anchoa,4.5)(\anchom,1.8)
        \qbezier(\anchom,1.8)(\anchob,-1.5)(\ancho,4)
       \end{picture} \\[-4pt]{#3}
                       \end{array} \kern-4pt    }

\newcommand{\lhat}[3][0]{\altura=0 \advance\altura by #1
           \ancho=#2 \anchom=\ancho \divide\anchom by 2
           \anchoa=\ancho \divide\anchoa by 4
           \anchob=\anchom \advance\anchob by \anchoa
           \kern-3pt \begin{array}[b]{c}
           \begin{picture}(1,1)(\anchom,-\altura)
        \qbezier(0,2)(\anchoa,4)(\anchom,6)
        \qbezier(\anchom,6)(\anchob,4)(\ancho,2)
        \qbezier(0,2)(\anchoa,3.8)(\anchom,5.6)
        \qbezier(\anchom,5.6)(\anchob,3.8)(\ancho,2)
       \end{picture} \\[-4pt] {#3}
                       \end{array} \kern-4pt    }

\makeatletter
\newcommand\prol{\@ifstar{\@proldf}{\@prolpf}}  
\def\@prolpf{\@ifnextchar[{\@prolpf@wrt}{\@prolpf@}}
\def\@prolpf@wrt[#1]#2{\@ifnextchar[{\@prolpf@wrt@at{#1}{#2}}{\@prolpf@wrt@{#1}{#2}}}
\def\@prolpf@wrt@at#1#2[#3]{\prolsymbol^{#1}_{#3}#2}
\def\@prolpf@wrt@#1#2{\prolsymbol^{#1}#2}
\def\@prolpf@#1{\@ifnextchar[{\@prolpf@at{#1}}{\@prolpf@@{#1}}}
\def\@prolpf@at#1[#2]{\prolsymbol_{#2}#1}
\def\@prolpf@@#1{\prolsymbol#1}
\def\@proldf{\@ifnextchar[{\@proldf@wrt}{\@proldf@}}
\def\@proldf@wrt[#1]#2{\@ifnextchar[{\@proldf@wrt@at{#1}{#2}}{\@proldf@wrt@{#1}{#2}}}
\def\@proldf@wrt@at#1#2[#3]{\prolsymbol^{*#1}_{#3}#2}
\def\@proldf@wrt@#1#2{\prolsymbol^{*#1}#2}
\def\@proldf@#1{\@ifnextchar[{\@proldf@at{#1}}{\@proldf@@{#1}}}
\def\@proldf@at#1[#2]{\prolsymbol^*_{#2}#1}
\def\@proldf@@#1{\prolsymbol^*#1}
\def\prolsymbol{\mathcal{T}}
\makeatother

\begin{document}

\title[Invariant measures  for contact Hamiltonian systems]{Invariant measures  for contact Hamiltonian systems: symplectic sandwiches with contact bread }

\author[A. Bravetti]{A. Bravetti}
\address{A. Bravetti:
Instituto de Investigaciones en Matem\'aticas Aplicadas y en Sistemas (IIMAS) \\ 
Universidad Nacional Aut\'onoma de M\'exico \\
 A.P. 70-543, 04510 Ciudad de M\'exico, M\'exico} 
\email{alessandro.bravetti@iimas.unam.mx}

\author[M. de Le\'on]{M. de Le\'on} 
\address{M. de Le\'on:
Instituto de Ciencias Matem\'aticas (CSIC-UAM-UC3M-UCM) and Real Academia de Ciencias\\
C/Nicol\'as Cabrera 13-15, 28049 Madrid, Spain}
\email{mdeleon@icmat.es}

\author[J.\ C.\ Marrero]{J.\ C.\ Marrero}
\address{J.\ C.\  Marrero:
ULL-CSIC Geometr\'{\i}a Diferencial y Mec\'anica Geom\'etrica\\
Departamento de Mate\-m\'a\-ti\-cas, Estad{\'\i}stica e IO, Secci\'on de
Ma\-te\-m\'a\-ti\-cas, Universidad de La Laguna \\
La Laguna, Tenerife, Canary Islands, Spain} 
\email{jcmarrer@ull.edu.es}		

\author[E. Padr\'on]{E. Padr\'on}
\address{E. Padr\'on:
ULL-CSIC Geometr{\'\i}a Diferencial y Mec\'anica Geom\'etrica\\
Departamento de Matem\'aticas, Estad{\'\i}stica e I O \\
Secci\'on de Matem\'aticas, Facultad de Ciencias \\
Universidad de la Laguna, La Laguna, Tenerife, Canary Islands, Spain}
\email{mepadron@ull.edu.es} 
		
%
%
%

 \thanks{\noindent {\it 2020 Mathematics Subject Classification}. 37A05; 37J55; 53D10; 70G45; 70H05}

 \keywords {\noindent Contact manifold, contact Hamiltonian system, Reeb dynamics, exact symplectic manifold, Liouville vector field, invariant measure}


\begin{abstract}

We prove that, under some natural conditions, Hamiltonian systems on a contact manifold $C$ 
can be split into a Reeb dynamics on an open subset of  $C$ and a Liouville dynamics on a submanifold of  $C$ of 
codimension $1.$ For the Reeb dynamics we find an invariant measure. 
Moreover,  we show that, under certain completeness conditions, the existence of an invariant measure for the 
Liouville dynamics can be characterized using the notion of a symplectic sandwich with contact bread.
\end{abstract}

\maketitle

\setcounter{section}{0}

\section{Introduction}
\setcounter{equation}{0}
Contact Hamiltonian systems have been the subject of intense study over the last years.
This is in part because {contact structures} have found applications in many areas of science, ranging e.g.~from 
{classical and quantum mechanics~\cite{AoP2017,CiCrMa,CiDiIbMaSc1,DeLeon2017,Gaset1}
to field theories~\cite{CiDiIbMaSc2,Gaset2}}, from statistical mechanics~\cite{BrTaPRE} to statistics~\cite{Mike}, and from optimization~\cite{optimization2019} 
to celestial mechanics, cosmology {and relativity~\cite{BaIbLa,CELE2020,Lazo2017,DSloan1,DSloan2}.}

One question that is of particular interest in the study of such systems is the existence of invariant measures for the flow.
This is undoubtedly one of the most important questions about dynamical systems in general. In fact, for a dynamical system with $n$ degrees of freedom, 
an invariant measure and $(n-2)$ first integrals which are independent on a level set, we have that the solutions of the system in the level set may 
be found by quadratures (see, for instance,~\cite{Ko1988} and the references therein). So, for this reason, invariant measures for interesting special dynamical 
systems have been intensively discussed in literature
(see e.g.~\cite{Ma2010} for Poisson Hamiltonian systems, and~\cite{FeGaMa2015} and the references therein for non-holonomic mechanical systems).

{Recently, some methods have been proposed for the integration of  contact Hamiltonian systems (see for example \cite{GP}). 
Invariant measures for contact Hamiltonian systems are a good tool in order to solve this problem.} 
Moreover, the existence of these  measures is also interesting because a positive answer could be used, in conjunction
with the geometric discretization methods proposed in~\cite{CELE2020,Simo2020,Ver2019}, to develop new algorithms for statistical physics
(e.g.~Nos\'e--Hoover-like methods for molecular dynamics~\cite{BrTaPRE}) and for statistics (e.g.~for using
contact Hamiltonian dynamics in the context of Hamiltonian Monte Carlo methods~\cite{Mike2,Mike}),
and it could be also relevant for the question of the existence of attractors in inflationary cosmology~\cite{DSloan1,DSloan2}. 

Contrary to standard symplectic Hamiltonian systems, for which Liouville's theorem states the existence of a natural invariant volume,
the situation for contact flows is more involved, as we shall see shortly. A partial answer to the existence of an invariant measure for
contact flows has been given in~\cite{BrTa}, where the authors have found an analogue of Liouville's theorem for contact flows
(cf.~Theorem~2 in~\cite{BrTa}). 
However, such theorem provides an invariant measure that is well-defined only in the region where the Hamiltonian is different 
from zero, whereas the corresponding statement in~\cite{BrTa} concerning the (invariant) region where the Hamiltonian is 
identically zero 
(cf.~Theorem~1 therein) is incorrect. Therefore it remains open the question of exploring conditions for the existence of an invariant measure
in the region where the contact Hamiltonian vanishes identically.

Motivated by these arguments, in this work we start an  analysis 
of the existence of invariant measures for contact flows. To do this, one may proceed directly or one may consider the symplectification of the contact structure and, then, discuss the problem in the setting of homogeneous symplectic manifolds. In fact, contact structures and, more generally, non-coorientable contact structures may be considered as homogeneous symplectic structures. Here, the homogeneity is associated with a principal $\R^{\times}$-bundle structure (a nice and detailed discussion of this fact may be found in \cite{BrGrGr}). Anyway, in this paper, we opt for the first option and we deal with the problem directly in the contact manifold.   

In fact, we proceed in two steps: 

\begin{itemize}

\item In the first step, we provide a suitable description of the contact Hamiltonian dynamics. 
In particular, we prove that, in the open subset $U$, where the Hamiltonian function $H$ has no zeros,  
the Hamiltonian vector field $X_H$ is the Reeb vector field of a conformal change of the original contact form $\eta$ (Theorem \ref{casoU}). 

On the other hand, if the Reeb vector field $\xi$ associated with $\eta$ is transverse to the zero level set $S$ of $H$, 
then we see that the differential of the restriction of $\eta$ to $S$ induces an exact symplectic structure on $S$ and 
the restriction of $X_H$ to $S$ is, up to reparametrization, the Liouville vector field of the exact symplectic manifold $S$ (Theorem \ref{casoS}).

\item In the second step, our strategy in order to discuss the existence of an invariant measure for $X_H$ will be 
to study it separately on the regions $U$ and $S$. Since $X_{H|U}$ is the Reeb vector field of a conformal change of 
the contact structure $\eta$, we directly find an invariant measure for $X_{H|U}$ and we recover the corresponding result in \cite{BrTa} (see Corollary \ref{BrTa}). 

On the other hand, the restriction of $X_H$ to $S$ does not admit, in general, an invariant measure. 
For instance, if $X_{H|S}$ has a critical point then there cannot be any invariant measure for $X_{H|S}$ (Corollary \ref{CorObstruction}). 
However under the assumptions that $\xi(H)=\gamma\in \mathbb R-\{0\}$ and that both $\xi$ and $X_{H|S}$ are complete, 
we prove in Theorem~\ref{th:main}, 
as the main result of this paper, that $X_{H|S}$ admits an invariant measure if and only if the original contact manifold $C$ is a 
symplectic sandwich with contact bread (see Definition~\ref{def:symplecticsandwich})
and $X_{H|S}$ admits a suitable global rectification in the sandwich. 
\end{itemize}


The previous results are related to the theory of convex hypersurfaces of contact manifolds (see~\cite{ElGr,Gi}) 
and the theory of Weinstein and convex contact structures (see~\cite{CiEl,ElGr,Sa,We}). 
In fact, explicit relations are presented in Remarks~\ref{Eliashberg-Gromov-Sackel} and~\ref{Weinstein} and in the last section of the paper (Section~\ref{conclusions-future-work}).
Our analysis is also similar to the ``Sasaki-K\"ahler sandwich'' appearing in Sasakian geometry~\cite{boyerbook}
and has some common points with~\cite{Zotev2007,Zotev2019}, where the 
author discusses the case of manifolds equipped with closed degenerate 2-forms admitting contact points.

The paper is organized as follows. In Section \ref{section2}, we review some results and constructions 
on Hamiltonian dynamics in contact manifolds which will be used in the rest of the paper.
 In Section~\ref{section3}, we will prove that, under natural conditions, the Hamiltonian dynamics on a contact manifold $C$ 
 is a Reeb dynamics in the open subset $U$ of $C$ where the Hamiltonian function is different from zero, plus a 
 Liouville dynamics on the invariant submanifold $S$ in which the Hamiltonian is identically zero. 
In Section \ref{section4}, we will discuss the existence of invariant measures for the Reeb and Liouville dynamics in $U$ and $S$, 
respectively. We will assume that the Reeb vector filed $\xi$ of $C$ and the restriction $X_{H|S}$ to $S$ of the Hamiltonian vector field $X_H$ are complete. 
The paper ends with an appendix where we will present the analogous results for the more general case when $\xi$ and $X_{H|S}$ are not necessarily complete. 

\medskip 

\noindent {\bf Acknowledgements:} Lemma \ref{s1} and its proof were proposed to one of the authors (JCM) by Jes\'us Alvarez-L\'opez. 
This result is one of the key points in order to obtain the relation between the symplectic sandwich with contact bread 
and the existence of an invariant measure for the restriction of the Hamiltonian vector field to the zero level set of the Hamiltonian function. 
The authors are very grateful for the invaluable comments of Jes\'us Alvarez-L\'opez. 
The authors are also grateful to Kevin Sackel for his comments on the relationship between our results and the theory of convex contact structures 
and to Charles Boyer
and Lin Wang for their valuable comments.
JCM  and EP are supported by  Ministerio de Ciencia e Innovaci\'on (Spain) under grants PGC2018-098265-B-C32.
MdL has been supported by Ministerio de Ciencia e Innovaci\'on (Spain) under grants 
  MTM2016-76702-P and ``Severo Ochoa Programme for Centres of Excellence'' in R\&D (SEV-2015-0554).
This research was partially supported by the University of La Laguna, that sponsored the visit of AB.

\section{Hamiltonian dynamics on contact manifolds}\label{section2}
In this section, we will review some results and constructions on Hamiltonian dynamics in contact manifolds 
(for more details, 
see for instance \cite{boyercompletely,AoP2017,DeLeon2019,Go,LiMa,Val1}).
 
 Let $\eta$ be a $1$-form  on a $(2n+1)$-dimensional manifold $C$. Then $\eta$ is \emph{a contact form} if the $(2n+1)$-form
 \begin{equation}\label{volume}
 \eta\wedge (d\eta)^n
 \end{equation}
 defines a volume form on $C.$   The couple $(C,\eta)$ is said to be \emph{a contact manifold}.  
 In such a case, we can consider \emph{the Reeb vector field} $\xi\in {\mathfrak X}(C)$ associated with $(C,\eta)$ which is  characterized by the conditions 
 \begin{equation}\label{Reeb}
 \eta(\xi)=1\mbox{ and } \iota_\xi d\eta=0.
 \end{equation}
 
An alternative characterization of a contact structure is the following: for a $1$-form  $\eta$ on $C,$ we can consider  
{\bf the vector bundle 
morphism}  
 $\flat_\eta:TC\to T^*C$ from the tangent bundle to the cotangent  bundle  of $C$ given by
 \begin{equation}\label{flat}
\flat_\eta(v_x)=\iota_{v_x}((d\eta)(x)) +\eta(x)(v_x)\eta(x),\; \; \mbox{ for } v_x \in T_xC \mbox{  with }x\in C.
 \end{equation}
 Then, $\eta$ is a contact $1$-form on $C$  if and only if $\flat_\eta$ is a vector bundle isomorphism. Note that if $(C,\eta)$ 
 is a contact manifold then, using (\ref{Reeb})  and (\ref{flat}), we deduce that $\flat_\eta({\xi})=\eta.$ 
Moreover, associated with the contact $1$-form $\eta$, we have {\bf the $2$-vector} $\Lambda_\eta$ on $C$ given by  
$$\Lambda_\eta(\alpha,\beta)=d\eta(\flat_{\eta}^{-1}(\alpha),\flat_{\eta}^{-1}(\beta)) \mbox{ with } \alpha,\beta\in \Omega^1(C).$$
 
On the other hand, it is well-known~\cite{Go} that for every point $x$ on a contact manifold $(C,\eta),$ there are {\bf Darboux coordinates} $(z,q^1,\dots, q^n,p_1,\dots, p_n)$ 
on an open neighborhood of $x$ such that the local expression of the contact $1$-form $\eta$ is 
 $$\eta=dz-\sum_{i=1}^n p_idq^i.$$
 Therefore, the local expressions of the Reeb vector field $\xi$ and  the $2$-vector $\Lambda_\eta$ in such coordinates are
  $$\xi=\frac{\partial }{\partial z}\quad \mbox{ and } \quad 
  \Lambda_\eta=\sum_{i=1}^n \left(\frac{\partial }{\partial q^i}+  p_i\frac{\partial }{\partial z}\right)\wedge \frac{\partial }{\partial p_i}.$$
 
 \begin{example}[{\bf Contactification versus symplectification}]\label{Ex-2.1} {\rm Suppose that $(M,d\lambda)$ is an exact symplectic manifold, 
 then $dz + \lambda$  defines a contact form on $\R\times M$. In this case, the Reeb vector 
 is $\displaystyle\frac{\partial }{\partial z}$  and the contact manifold $(\R\times M, dz + \lambda)$ is called {\bf the contactification of $(M,d\lambda)$}.
 
 {On the other hand, if $\eta$ is a $1$-form on $C$ and $c\in \R-\{0\}$, then one may prove that the $2$-form $\Omega_c=\exp(cs)(d\eta + cds\wedge \eta)$ 
 is a symplectic structure on $\R\times C$  if and only if $\eta$ is a contact structure on $C.$ In the particular case when $c=-1$, 
 the symplectic manifold $(\R\times C,\Omega_1)$ is said to be {\bf the symplectification of the contact manifold $(C,\eta)$}~\cite{Iba1997}.
 }}
\end{example}

Using the $2$-vector $\Lambda_\eta$ and the Reeb vector field $\xi,$ we can introduce a bracket $\{\cdot,\cdot\}_\eta: C^\infty(C)\times C^\infty(C)\to C^\infty(C)$ 
  on the space  of differentiable functions on $C$ defined as 
$$\{f,g\}_\eta=\Lambda_\eta(df,dg)+f\xi(g) - g\xi(f),\;\;\;\forall f,g\in C^\infty(C).$$
Note that $\{\cdot,-1\}=\xi.$ This bracket is not in general a Poisson bracket because it is not a  derivation in each argument 
with respect to the standard product of real functions. In fact, $(C^\infty(C),\{\cdot, \cdot\}_\eta)$ is a Lie algebra and, for all $H\in C^\infty(C)$,  $\{\cdot,H\}_\eta$ defines a first order differential operator, i.e.
$$\{f f',H\}_\eta=f\{f',H\}_\eta + f'\{f,H\}_\eta -ff'\{1,H\}_\eta,\;\; \forall f,f'\in C^\infty(C).$$
Therefore  $\{\cdot,\cdot\}_\eta$ is a {\bf Jacobi bracket} (see \cite{GuLi,Ki,Li}).
Moreover, the operator $\{\cdot,H\}_\eta:C^\infty(C)\to C^\infty(C)$ can be identified with a couple $(X_H,f_H),$ where $X_H$ is a vector field and $f_H$ is a function on $C,$ which are related 
with $\{\cdot, H\}_{\eta}$ by the formula 
$$\{f,H\}_\eta=X_H(f) + f_Hf.$$
A direct computation proves that
\begin{equation}\label{XH}
X_H=-\iota_{dH}\Lambda_\eta - H\xi\mbox{ and }f_H=\xi(H).
\end{equation}

The vector field $X_H$ is called \emph{the Hamiltonian vector field associated with $H.$}
In fact, the Hamiltonian vector field $X_H$ is the unique vector field on $C$ characterized by 
\begin{equation}\label{cont}
\iota_{X_H}d\eta=dH-\xi(H)\eta \quad \mbox{ and } \quad \iota_{X_H}\eta=-H.
\end{equation} 
The integral curves of $X_H$ are the solutions of \emph{the contact Hamiltonian dynamical system $(C,\eta,H)$}.  
It is clear that the Reeb vector field  is the Hamiltonian vector field associated with the constant function $-1.$ 
   
Using Darboux local coordinates $(z,q^i,p_i)$ on $C,$ and Einstein's notation implying a sum over repeated indices,
the Hamiltonian vector field $X_H$ is just 
 $$X_H=\frac{\partial H}{\partial p_i}\frac{\partial }{\partial q^i}-\left(\frac{\partial H}{\partial q^i}+ p_i\frac{\partial H}{\partial z}\right)\frac{\partial }{\partial p_i} + \left(p_i\frac{\partial H}{\partial p_i}-H\right)\frac{\partial }{\partial z}$$
 and the contact Hamiltonian equations associated with $X_H,$ are given by 
 $$\frac{dq^i}{dt}=\frac{\partial H}{\partial p_i},\;\;\; \frac{dp_i}{dt}=-\left(\frac{\partial H}{\partial q^i}+ p_i\frac{\partial H}{\partial z}\right),\;\;  \frac{dz}{dt}=p_i\frac{\partial H}{\partial p_i}-H.$$
 
 The nature of contact Hamiltonian dynamics is different from the nature of symplectic Hamiltonian dynamics. 
 In fact, for a contact Hamiltonian system $(C, \eta, H)$ we have the following results: 
 \begin{itemize}
\item  $H$ is not a first integral of $X_H$. Indeed, using (\ref{XH}), we deduce
 \begin{equation}\label{HH}
 X_H(H)=-H\xi(H).
 \end{equation}
 \item From (\ref{cont}), it follows that 
 \begin{equation}\label{Lie-derivative-eta}
 {\mathcal L}_{X_H}\eta=\iota_{X_H}d\eta + d(\iota_{X_H}\eta) = -\xi(H)\eta.
 \end{equation} 
 Therefore, the contact form $\eta$ is not invariant under the action of $X_H.$
 \item The Liouville volume $\nu=\eta\wedge (d\eta)^n$ is not invariant under the action of $X_H$. In fact, using (\ref{Lie-derivative-eta}), we have that
\begin{equation}\label{noinv}
{\mathcal L}_{X_H}\nu=-\xi(H) \nu + n \eta \wedge d({\mathcal L}_X\eta) \wedge (d\eta)^{n-1} = -(n+1)\xi(H)\nu.
\end{equation}
 \end{itemize}

However, $\eta$ and $\nu$ are invariant under the action of $\xi=X_{-1},$ i.e.
$${\mathcal L}_{\xi}\eta=0 \mbox{ and } {\mathcal L}_{\xi}\nu=0.$$

Next, we will present some examples of interesting contact Hamiltonian systems.

 \begin{example}[{\bf Dissipative mechanical systems}]\label{ex:diss1}
 {\rm 
   We consider the product manifold $C=\R\times T^*\R^2$ endowed with the canonical contact structure 
   $\eta = dz - p_1dq^1-p_2dq^2,$ where $z$ and $({\bf q},{\bf p})=(q^1,q^2,p_1,p_2)$ are global coordinates on $\R$ and $T^*\R^2,$ respectively.   The contact Hamiltonian function $H: C \to \R$ is given by 
\[
H(z,{\bf q},{\bf p})=\frac{1}{2}(p_1^2+p_2^2) + V({\bf q})+\gamma z,\;\;\; \mbox{ with } \gamma\in \R-\{0\}. 
\]
So, we have that  the contact Hamiltonian equations are  
  $$\frac{dq^i}{dt}=p_i,\;\;\; \frac{dp_i}{dt}=-\frac{\partial V}{\partial q^i}-\gamma p_i, \quad \mbox{ with } \quad i=1,2 \quad\mbox{ and }\quad\frac{dz}{dt}=(p_1^2+p_2^2)-H$$
or, equivalently,  
 $$\ddot{q}^i + \gamma\dot{q}^i +  \frac{\partial V}{\partial q^i}=0\quad \mbox{ with }\quad i=1,2, \quad\mbox{ and }\quad \frac{dz}{dt}=\frac{1}{2}((\dot q^1)^2 + (\dot q^2)^2)-V({\bf q}) -\gamma z$$
 The reader may recognize the typical second order differential equation associated with 
 a dissipative mechanical system, where the dissipation is linear in the velocity 
 (see e.g.~\cite{AoP2017, DeLeon2017,Gaset1} for further details).}
 \end{example}

 \begin{example} [{\bf Reeb dynamics and standard Hamiltonian systems}] 
{\rm Let $Q$ be a smooth manifold and $\omega_Q = -d\lambda_Q$ the canonical symplectic structure on $T^*Q$, 
with $\lambda_Q$ the Liouville $1$-form on $T^*Q$. Suppose that $H: T^*Q \to \R$ is a Hamiltonian function and ${\mathcal H}_H \in \frak{X}(T^*Q)$ 
is the symplectic Hamiltonian vector field associated with $H$. This means that ${\mathcal H}_H$ is defined
by the following condition
\[
\iota_{{\mathcal H}_H}\omega_Q = dH.
\] 
 Denote by $\Delta$  the Liouville vector field on $T^*Q$, which is defined by 
 $$\iota_\Delta\omega_Q=-\lambda_Q.$$
 Let $c\in \R$ be a real number such that  $\Delta(H)\not=0$ on $H^{-1}(c).$ 
 Then $C=H^{-1}(c)$ is a  submanifold  of $T^*Q$ of codimension $1$  
 (note that the condition $\Delta(H)\not=0$  on $H^{-1}(c)$ implies that $H$ is regular at every point of $H^{-1}(c)$). 
 Moreover, using that $H$ is a first integral of ${\mathcal H}_H$, it follows that the restriction of ${\mathcal H}_H$ to $C$ is tangent to $C$. 
 In addition, since $\lambda_Q({\mathcal H}_H) = \Delta(H)$, we have that $\left(\lambda_Q({\mathcal H}_H)\right)_{|C} \neq 0$, at every point of $C$. 
 Thus, we may use Theorem 5.9 in \cite{Go} and deduce the following results: 
\begin{itemize} 
\item
The $1$-form $\eta=i_C^*\lambda_Q$ is a contact form on $C, $ where   $i_C:C\to T^*Q$ is the canonical inclusion. 
\item
The Reeb vector field of the contact manifold $(C, \eta)$ is the vector field $\xi$ on $C$ given by      
 $$\xi=\frac{({\mathcal H}_{H})_{|C}}{\Delta(H)_{|C}}.$$
%
\end{itemize}
So, in conclusion, the restriction to $C$ of the symplectic
Hamiltonian dynamics ${\mathcal H}_{H}$ is just a re-parametrization of the contact Hamiltonian vector field on $C$ corresponding to the constant function $-1$,
that is, the Reeb dynamics.}
 \end{example}
 
 \section{Contact Hamiltonian systems and Reeb-Liouville dynamics}\label{section3}
 \label{Section3}

Let $(C,\eta, H)$ be a contact Hamiltonian system on a manifold $C$ of dimension $2n+1$. 
Denote by $U$  the open subset of $C$
 $$U=\{x\in C/H(x)\not=0\}$$
and by $i_U:U\to C$ the canonical inclusion. Then, in this section, we will prove the following facts:
\begin{itemize}
\item
It is possible to define another contact form on $U$ 
so that the restriction of the Hamiltonian vector field $X_H$ to $U$ is just the Reeb vector field of the new contact form. 
\item
Under reasonable hypotheses, one has that the complementary subset $S = C-U$ of $U$ is a submanifold of codimension $1$ of $C$ which admits an exact symplectic structure.
\item
The restriction of $X_H$ to $S$ is tangent to $S$ and $X_{H|S}$ is a reparametrization of the Liouville vector field of the exact symplectic manifold $S$. 
\end{itemize} 
In conclusion, the contact Hamiltonian dynamics is confined to the two complementary regions $U$ and $S$, 
and it is a Reeb dynamics on $U$ and a Liouville dynamics on $S$.

Next, we will prove the first result.
\begin{theorem}\label{casoU} The $1$-form  
$$\eta_H=-\displaystyle\frac{1}{H\circ i_U}i_U^*\eta$$
 defines a contact structure on $U$. Moreover,  the Reeb vector field of $\eta_{U}$ 
 is just the restriction $X_{H|U}$ of the Hamiltonian vector field associated with $H$ to $U$. 
\end{theorem}
\begin{proof} A direct computation proves that 
\begin{equation}\label{d-eta-H}
d\eta_H=\frac{1}{(H\circ i_U)^2}d(H\circ i_U)\wedge i_U^*\eta - \frac{1}{H\circ i_U} i_U^*d\eta.
\end{equation}
Therefore, 
\begin{equation}\label{Conformal-change-Liouville-volume}
\eta_H\wedge (d\eta_H)^n = \frac{(-1)^{n+1}}{(H \circ i_U)^{n+1}} i_U^*(\eta \wedge (d\eta)^n),
\end{equation}
is a volume form on $U$ and hence 
$\eta_H$ is a contact $1$-form on $U$. Furthermore, from (\ref{cont}), it follows that  
$$\eta_H(X_{H|U})=- \frac{1}{H\circ i_U}(\eta(X_H)\circ i_U)=1.$$
On the other hand, if $\xi$ is the Reeb vector field of $C$ then, using (\ref{cont}), (\ref{HH}) and (\ref{d-eta-H}), we deduce that
\[
\iota_{X_{H|U}}d\eta_H = -\frac{1}{H \circ i_U} (\xi(H) \circ i_U) i_U^*\eta + \frac{1}{H \circ i_U}d(H \circ i_U) - \frac{1}{H \circ i_U}(d(H \circ i_U) - (\xi(H) \circ i_U) i_U^*\eta) = 0.
\]
This proves the result.
\end{proof}

Now, a natural question arises: 
	what happens in the subset $S=C-U=H^{-1}(0)$? 
In order to give an answer to this question, we start with a preliminary result, which 
is a consequence from (\ref{HH}).
\begin{lemma}\label{0-regular-value}
If $S\not=\emptyset$ and $0$ is a regular value of $H$, then $S$ is a submanifold of $C$ of codimension~1 
which is invariant under the action of $X_H$, that is, $X_{H|S}$ is tangent to $S.$ 
\end{lemma}

After Lemma~\ref{0-regular-value}, two natural questions arise: 
\begin{enumerate}
\item What is the induced geometric structure on $S$?
\item What is the description of $X_{H|S}$ in terms of the previous structure?
\end{enumerate}
Related to these two questions, we may prove the following result.

\begin{theorem}\label{casoS}
Let $H: C \to \R$ be a Hamiltonian function on a contact manifold $(C, \eta)$ with Reeb vector field $\xi$ 
and assume that $S = H^{-1}(0) \not=\emptyset$ and that $\xi(H)(x)\not=0,$ for all $x\in S.$ 
\begin{enumerate}
\item If $\theta = i_S^*\eta$ on $S$, with $i_S:S\to C$ the inclusion map, then $\Omega=d\theta$ is an exact symplectic structure on $S$. 
\item If $\Delta$ is the Liouville vector field of the exact symplectic manifold $(S,\Omega)$, i.e.~$\Delta$ is the vector field characterized by 
\begin{equation}\label{Liouville}
\iota_\Delta\Omega=\theta,
\end{equation}
then $X_{H|S}$ is the reparametrization of $\Delta$ given by 
\begin{equation}\label{X-H-S-Delta}
X_{H|S}=-(\xi(H)\circ i_S)\Delta.
\end{equation}
\end{enumerate}
\end{theorem}
\begin{proof} The condition $\xi(H)(x)\not=0$ for all $x\in S$ implies that $0$ is a regular value of $H$ and, using Lemma~\ref{0-regular-value}, 
we deduce that $S$ is a submanifold of $C$ of codimension $1$ and $X_{H|S}$ is a vector field on $S$.

(i) We will see that $\Omega$ is a non-degenerate $2$-form on $S$. In fact, if $x\in S$ and $v \in T_xS$ satisfies
\[
\iota_v\Omega(x) = 0,
\]
we have that $(\iota_v d\eta(x))(u) = 0$ for all $u \in T_xS$. Now, since $\xi$ is transverse to the submanifold $S$, it follows that
\begin{equation}\label{Reeb-transverse}
T_xC = T_xS \oplus \langle \xi(x) \rangle.
\end{equation}
However, $\iota_v d\eta(x) = 0$ implies $v \in \langle \xi(x)\rangle$, and hence $v \in \langle \xi(x) \rangle \cap T_xS$ which, by (\ref{Reeb-transverse}), means that $v = 0$.

(ii) Using (\ref{cont}) and the fact that $i_S^*(dH) = 0$, we obtain that
\[
\iota_{X_{H|S}} \Omega =i_S^*(dH - \xi(H)\eta)=-(\xi(H)\circ i_S) \theta=\iota_{-(\xi(H)\circ i_S )\Delta}\Omega.
\]
Therefore, since $\Omega$ is a non-degenerate $2$-form, we conclude that $X_{H|S} = -(\xi(H)\circ i_S )\Delta$.
\end{proof}
\begin{remark}
{\rm Note that if $\xi(H)\equiv 0$ then the behavior of $X_H$ is similar to that of a Hamiltonian vector field in an exact 
presymplectic manifold $(C^{2n+1},d\eta)$  of dimension $2n+1$ (of rank $2n$). In fact, from (\ref{HH}), (\ref{Lie-derivative-eta}) and (\ref{noinv}), we deduce that:
\begin{itemize}
\item
$H$ is a first integral of $X_H$ and
\item
The contact structure $\eta$ and the Liouville volume $\nu = \eta \wedge (d\eta)^n$ are $X_H$-invariant.
\end{itemize}
}
\end{remark}
\begin{remark}\label{Eliashberg-Gromov-Sackel}
{\rm Under the same hypotheses as in Theorem \ref{casoS}, we have that:
\begin{itemize}
\item
$S = H^{-1}(0)$ is a convex hypersurface of the contact manifold $(C, \eta)$ (see \cite{ElGr,Gi} for the theory of convex hypersurfaces).
\item
If, in addition, we assume that $\xi(H)$ is a positive function on $C$, 
then it is easy to prove that there exists a Riemannian metric $g$ on $C$ such that $\xi(H) = \| \xi \|^2$ and 
$dH$ is the dual metric of $\xi$ with respect to $g$. So, if we take $0 < \delta \leq \displaystyle \frac{1}{2}$, it follows that
\[
\xi(H) \geq \delta ( \| \xi \|^2 + \| dH \|^2).
\]
This implies that $(\xi, H)$ is a gradient-like pair and that $(C, \eta, \xi, H)$ is a convex contact manifold (without singular points) in the terminology by Sackel (see \cite{Sa}). 
\end{itemize}
}
\end{remark}
\section{Invariant measures for contact Hamiltonian dynamics}
\label{section4}
Let $(C, \eta, H)$ be a Hamiltonian system on a connected contact manifold $C$ of dimension $2n+1$. 
Then, $C$ is orientable and a (positive) volume form on $C$ is just the Liouville volume $\nu = \eta \wedge (d\eta)^n$. On the other hand, any positive volume form on $C$ is given by
\[
\nu_{\sigma} = \exp(\sigma) \nu = \exp(\sigma)\eta \wedge (d\eta)^n,
\]
with $\sigma: C \to \R \in C^{\infty}(C)$.
\begin{proposition}\label{propcom} 
The volume form 
$\nu_\sigma $ is an invariat measure for the Hamiltonian vector field $X_H$ if and only if
\begin{equation}\label{com}
X_H(\sigma)=(n+1)\xi(H),
\end{equation}
where $\xi$ is the Reeb vector field associated with $(C,\eta).$ 
\end{proposition}
\begin{proof} The volume form $\nu_{\sigma}$ is an invariant measure for $X_H$ if and only if ${\mathcal L}_{X_H}\nu_{\sigma} = 0$.

Now, using (\ref{noinv}), we deduce that 
$${\mathcal L}_{X_H}\nu_\sigma=X_H(\sigma)\nu_\sigma + \exp(\sigma) {\mathcal L}_{X_{H}}\nu = \left(X_H(\sigma) - (n+1)\xi(H)\right)\nu_\sigma,$$
which implies the result.
\end{proof}

If $(z, q^{i}, p_i)$ are Darboux coordinates for the contact manifold $(C,\eta),$ we have that the local expression of Equation~\eqref{com} is 
$$\frac{\partial H}{\partial p_i}\frac{\partial \sigma}{\partial q^i}-
\left(\frac{\partial H}{\partial q^{i}}+ p_i\frac{\partial H}{\partial z}\right)\frac{\partial \sigma}{\partial p_i}+
\left(p_i\frac{\partial H}{\partial p_i}-H\right)\frac{\partial \sigma}{\partial z}=(n+1)\frac{\partial H}{\partial z}$$ for which finding an explicit form for the solutions seems to be rather prohibitive.

So, our strategy in order to discuss the existence of invariant measures for the contact Hamiltonian dynamics will be to study it separately on the regions 
\[
U = \{x \in C / H(x) \not= 0\} \; \; \mbox{ and } \; \; S = C - U.
\]

Since the restriction 
 $X_{H|U}$ 
 is just the Reeb vector field of the contact form $\eta_H=-\displaystyle\frac{1}{H\circ i_U}i_U^*\eta$ (see Theorem~\ref{casoU}), 
 then  the Liouville volume induced by $\eta_H$  is invariant under the action of  $X_{H|U}$, that is
 $${\mathcal L}_{X_{H|U}}(\eta_H\wedge (d\eta_H)^n)=0.$$
 
Therefore, using (\ref{Conformal-change-Liouville-volume}), we deduce the following result which was proved in \cite{BrTa}. 
%
\begin{corollary}
\label{BrTa}
The volume form $\displaystyle\frac{1}{H^{n+1}}\eta\wedge (d\eta)^{n}$ on $U$ is an invariant measure for $X_{H|U}$,
which corresponds to the function $\sigma =-(n+1)\ln(H)$ in Proposition~\ref{propcom}.
\end{corollary}

Now,  in order to discuss the existence of an invariant measure on $S = H^{-1}(0) = C - U$, we will assume the hypothesis   
$$\xi(H)(x)\not=0\quad \mbox{ for all } x\in S.$$
Then, using Theorem~\ref{casoS}, we have that the restriction of $X_H$ to the closed submanifold $S$ is tangent to $S$ and 
 $$X_{H|S}=-(\xi(H)\circ i_S)\Delta, $$
 where  $i_S:S\to C$ is the inclusion map and $\Delta$ denotes  the Liouville vector field of the exact symplectic manifold $(S,di_S^*\eta)$. 
 
 Moreover, we may prove that the existence of an invariant measure for the contact Hamiltonian dynamics on $S$ is equivalent to the existence of an invariant measure for the Liouville vector field.
 
 \begin{lemma}\label{XD} $\Delta$ preserves a volume form $\nu_S$ on $S$ if and only if $X_{H|S}$ preserves the volume form $\displaystyle \frac{1}{\xi(H)\circ i_S}\nu_S$. 
 \end{lemma}
 \begin{proof} Indeed, using the properties of the Lie derivative, we deduce that
\begin{equation}\label{divergence-1}
{\mathcal L}_{\Delta}\nu_S=-{\mathcal L}_{\frac{X_{H|S}}{\xi(H)\circ i_S}}\nu_S
=-d\left(\frac{1}{\xi(H) \circ i_S}\right)\wedge \iota_{X_{H|S}} \nu_S -\frac{1}{\xi(H)\circ i_S} {\mathcal L}_{X_{H|S}}\nu_S,
\end{equation}
and
\begin{equation}\label{divergence-2}
{\mathcal L}_{X_{H|S}}\left(\frac{1}{\xi(H) \circ i_S}\nu_S\right) = X_{H|S}\left(\frac{1}{\xi(H) \circ i_S)}\right) \nu_S + \frac{1}{\xi(H) \circ i_S}{\mathcal L}_{X_{H|S}}\nu_S.
\end{equation}
Now, since $d\left(\displaystyle \frac{1}{\xi(H) \circ i_S}\right) \wedge \nu_S$ is a $(2n+1)$-form on $S$, we have that $0 = d\left(\displaystyle \frac{1}{\xi(H) \circ i_S}\right) \wedge \nu_S$ and thus
\[
0 = \iota_{X_{H|S}}\left(d\left(\frac{1}{\xi(H) \circ i_S}\right) \wedge \nu_S\right) 
= X_{H|S} \left(\frac{1}{\xi(H) \circ i_S}\right) \nu_S - d\left(\frac{1}{\xi(H) \circ i_S}\right) \wedge \iota_{X_{H|S}} \nu_S,
\]
that is, 
\[
X_{H|S} \left(\frac{1}{\xi(H) \circ i_S}\right) \nu_S = d\left(\frac{1}{\xi(H) \circ i_S}\right) \wedge \iota_{X_{H|S}} \nu_S\,.
\]
So, using (\ref{divergence-1}) and (\ref{divergence-2}), we conclude that
\[
{\mathcal L}_{\Delta}\nu_S = - {\mathcal L}_{X_{H|S}}\left(\frac{1}{\xi(H) \circ i_S}\nu_S\right)
\]
which proves the result.
 \end{proof}
Therefore, our problem reduces to discussing the existence of an invariant measure for the 
Liouville vector field $\Delta$ of the exact symplectic manifold $(S,di_S^*\eta)$. Using this fact, we may deduce the following result.
\begin{theorem} \label{inv-measure-S}
Let $(C, \eta, H)$ be a contact Hamiltonian system on a contact manifold $(C, \eta)$ of dimension $2n+1$ with Reeb vector field $\xi$. 
If $\xi(H)(x) \not= 0$ for every $x \in S = H^{-1}(0)$, then the restriction to $S$ of the contact Hamiltonian dynamics $X_H$ admits an invariant measure 
if and only if there exists a real $C^{\infty}$-function $\sigma$ on $S$ such that
\begin{equation}\label{eq:conditiononS}
X_{H|S}(\sigma) = n(\xi(H) \circ i_S),
\end{equation}
where $i_S: S \to C$ is the canonical inclusion. 
Moreover, if the previous condition holds then 
\[
\displaystyle \frac{\exp(\sigma)}{\xi(H) \circ i_S}(d i_S^*\eta)^n
\] 
is an invariant measure for $X_{H|S}$.  
\end{theorem}
\begin{proof}
From the first item in Theorem \ref{casoS}, we have that $(S, di_S^*\eta)$ is an exact symplectic manifold of dimension $2n$ and a (positive) volume form $\nu_S$ on $S$ is given by
\begin{equation}\label{nu-S}
\nu_S = \exp(\sigma)(di_S^*\eta)^n,
\end{equation}
with $\sigma: S \to \R \in C^{\infty}(S)$.

Now, if $\Delta$ is the Liouville vector field on $S$ then, using the fact that $\iota_\Delta d(i_S^*\eta) = i_S^*\eta$, we deduce that
\[
{\mathcal L}_\Delta d(i_S^*\eta) = d {\mathcal L}_\Delta(i_S^*\eta) = d(\iota_\Delta d(i_S^*\eta) + d((i_S^*\eta)(\Delta))) = di_S^*\eta.
\]

Thus, using (\ref{nu-S}) and the properties of the Lie derivative, we obtain that
\[
{\mathcal L}_\Delta \nu_S = (\Delta(\sigma) + n)\nu_S
\]
and, therefore,
\begin{equation}\label{inv-volume-Delta}
{\mathcal L}_\Delta\nu_S = 0\quad \Leftrightarrow \quad\Delta(\sigma) = -n.
\end{equation}
On the other hand, from Lemma \ref{XD}, it follows that
\begin{equation}\label{using-XD}
{\mathcal L}_{X_{H|S}}\left(\frac{1}{\xi(H) \circ i_S}\, \nu_S\right) = 0\quad \Leftrightarrow\quad {\mathcal L}_\Delta\nu_S = 0.
\end{equation}
In addition, since $X_{H|S} = -(\xi(H) \circ i_S) \Delta$, we deduce that
\begin{equation}\label{derivative-density}
X_{H|S}(\sigma) = n (\xi(H) \circ i_S) \quad\Leftrightarrow\quad \Delta(\sigma) = -n.
\end{equation}
Thus, the result follows using (\ref{nu-S}), (\ref{inv-volume-Delta}), (\ref{using-XD}) and (\ref{derivative-density}).
\end{proof} 
%
%
An immediate consequence of the previous theorem is the following obstruction to the existence of a (global) invariant measure
on $S$.
\begin{corollary}\label{CorObstruction}
Let $(C, \eta, H)$ be a contact Hamiltonian system on a contact manifold $(C, \eta)$ of dimension $2n+1$ with Reeb vector field $\xi$. 
If $\xi(H)(x) \not= 0$ for every $x \in S = H^{-1}(0)$ and $X_{H|S}$ has a critical point, then there cannot exist any
invariant measure for  $X_{H|S}$. 
\end{corollary}

\begin{proof} If $x_0\in S$ is a critical point of $X_{H|S}$ and $\sigma$ is an arbitrary real $C^{\infty}$-function on $S$ then $X_{H}(\sigma)(x_0) = 0$. 
This, using Theorem \ref{inv-measure-S} and the fact $\xi(H)(x_0) \not= 0$, proves the result.
\end{proof}
After Corollary \ref{CorObstruction}, a natural question arises: what happens if $X_H(x) \not= 0$ 
for every point $x\in S$ or, equivalently (see (\ref{X-H-S-Delta})), if $\Delta(x) \not= 0$ for every $x \in S$?

The previous question may be reformulated as follows. Suppose that $(S, d\theta)$ is an exact 
symplectic manifold, that $\Delta$ is the Liouville vector field of $S$ and that $\Delta(x) \not= 0$ $\forall x \in S$. 
Then, our question is:

Is there a real $C^\infty$-function $\sigma$ on $S$ such that $\Delta(\sigma) = -n$?

\noindent 
It is clear that the previous question is equivalent to the following one:  

Is there a real $C^\infty$-function $\sigma$ on $S$ such that $\Delta(\sigma) = 1$?

\noindent Next, we will consider two particular examples. 

\begin{example}[\bf The cotangent bundle without the zero section]\label{ExCotangent}
{\rm Suppose that $Q$ is a smooth manifold of dimension $n$ and that our symplectic manifold is the open subset $S$ 
of the cotangent bundle $T^*Q$ of $Q$ given by $S = T^*Q - 0_Q(Q)$, where $0_Q: Q \to T^*Q$ is the zero section. 
On $S$, we consider the restriction of the symplectic structure $\omega_Q = -d\lambda_Q$ on $T^*Q$, with $\lambda_Q$ the Liouville $1$-form. 
If $\Delta$ is the Liouville vector field on $S$ and $(q^{i}, p_i)$ are fibred coordinates on $S$, we have that
\begin{equation}\label{Liouville-symplectic}
\lambda_Q = p_i dq^{i}, \; \; \omega_Q = dq^i \wedge dp_{i}, \; \; \Delta = p_i\frac{\partial}{\partial p_i}.
\end{equation}
So, it is clear that $\Delta(x) \not= 0$, for every $x \in S$.

Now, we take a Riemannian metric $g$ on $Q$,
\[
g = g_{ij}(q)dq^{i} \otimes dq^{j}, \; \; \mbox{ with } g_{ij} = g\left(\frac{\partial}{\partial q^{i}}, \frac{\partial}{\partial q^{j}}\right),
\]
and the kinetic energy $K_g: T^*Q \to \R$ associated with $g$ given by
\[
K_g(\alpha) = \frac{1}{2} \| \alpha \|^2, \; \; \mbox{ for } \alpha \in T^*Q.
\]
If $(g^{ij}(q))$ is the inverse matrix of $(g_{ij}(q))$, we have that
\begin{equation}\label{kinetic-energy}
K_g(q, p) = \frac{1}{2} g^{ij}(q)p_ip_j. 
\end{equation}
From (\ref{Liouville-symplectic}) and (\ref{kinetic-energy}), it follows that
\[
\Delta((K_g)_{|S}) = 2(K_g)_{|S}.
\]
Thus, if $\sigma: S \to \R$ is the real function on $S$ defined by $\sigma = \displaystyle \frac{1}{2} \ln (K_g)_{|S} + (F\circ \pi_Q)_{|S}$, with $F:Q\to {\R}$ 
a function on $Q$ and $\pi_Q:T^*Q\to Q$  the canonical projection, we deduce that
\[
\Delta(\sigma) = 1.
\]
Note that, in fact, the volume form $\nu_S$ on $S$ given by
\[
\nu_S = \exp(-n\sigma) j_S^*(d\lambda_Q)^n = ((K_g)_{|S}^{-\frac{n}{2}}\exp(-n(F\circ \pi_Q)_{|S})j_S^*(d\lambda_Q)^n
\]
is an invariant measure for $\Delta$, with $j_S: S \to T^*Q$ the canonical inclusion. In particular, if $F(q)=-\frac{1}{n}\ln(\rho(q))$, with $\rho(q)$ some distribution on $Q$, we get
 $$\nu_S =((K_g)_{|S}^{-\frac{n}{2}}(\rho\circ \pi_Q)_{|S})j_S^*(d\lambda_Q)^n\,,$$
which provides an invariant measure of the type needed in Hamiltonian Monte Carlo algorithms~\cite{Mike2}.
}
\end{example}

\begin{example}[\bf Dissipative mechanical systems (revisited)]\label{ExDiss2}
{\rm Consider the framework of Example~\ref{ex:diss1}.
If $\xi$ is the Reeb vector field of the canonical contact structure on $C$, we have that   
$$\xi(H) = \displaystyle\frac{\partial H}{\partial z} = \gamma \neq 0\,.$$ 
For this system, the closed submanifold $S$ of $C$ is 
\[ 
S = H^{-1}(0)= \left\{ \left(-\frac{1}{\gamma}\left(\frac{1}{2} (p_1^2 + p_2^2) + V({\bf q})\right), {\bf q}, {\bf p}\right) \in C / ({\bf q}, {\bf p}) \in T^*\mathbb{R}^2\right\} \simeq T^*\mathbb{R}^2, 
\]
and the $1$-form $\theta$ and the Liouville vector field $\Delta$ are given by 
\[
\theta =i_S^*\eta=-\sum_{i=1}^2 \left[\left(\frac{1}{\gamma} \frac{\partial V}{\partial q^i} + p_i\right)dq^i + \frac{1}{\gamma}p_idp_i\right]
\]
 and 
\[
X_{H|S}=-\gamma \Delta =  \sum_{i=1}^2 \left[p_i \frac{\partial}{\partial q^i} - \left(\gamma p_i + \frac{\partial V}{\partial q^i}\right) \frac{\partial}{\partial p_i}\right].
\]
Thus, if ${\bf p}\neq 0$ then $X_{H|S}({\bf q}, {\bf p}) \neq 0.$ On the other hand, $X_{H|S}({\bf q}, {\bf 0}) =0$ if and only if  ${\bf q}$ is a critical point of $V$. So, 
if $V$ has a critical point  then there cannot be  an invariant measure for $X_{H|S}.$

Now, we consider the particular case when 
$V({\bf q}) = q^1+q^2$. Then, $V$ has no critical points and, using Theorem  \ref{inv-measure-S}, 
we have that 

\begin{equation}\label{XHSex}
X_{H|S} =  \sum_{i=1}^2 \left[p_i \frac{\partial}{\partial q^i} - \left(\gamma p_i + 1\right) \frac{\partial}{\partial p_i}\right].
\end{equation}

admits an invariant measure if and only if there exists a 
real $C^\infty$-function $\sigma$ on $S$ such that $X_{H|S}(\sigma)=2\gamma$ or, equivalently 
 $$ p_1 \frac{\partial \sigma}{\partial q^1}  +p_2 \frac{\partial \sigma}{\partial q^2}- (\gamma p_1+1)\frac{\partial \sigma}{\partial p_1}-(\gamma p_2 +1)\frac{\partial \sigma}{\partial p_2} =2 \gamma.$$
A solution of this equation is 
\begin{equation}\label{sigma}
\sigma{\bf (q}, {\bf p}) =- \gamma (p_1 + p_2) - \gamma^2 (q^1 + q^2)
\end{equation}
Thus (see Theorem  \ref{inv-measure-S}), 
\[
\frac{1}{\gamma}\,\exp(-\gamma (p_1 + p_2) - \gamma^2 (q^1 + q^2)) dq^1\wedge dp_1 \wedge dq^2\wedge dp_2\]
is an invariant measure for  $X_{H|S}.$

}

\end{example}

{Next we} will show that,  in a situation like the previous example (that is, $\xi(H)=\gamma\not=0$),  
the existence of a function $\sigma: S\to \R$ such that $X_{H|S}(\sigma)=n\gamma$ is equivalent to  a certain 
trivialization $\R\times B$  of $S$ under which  $X_{H|S}$ is a reparametrization of $\displaystyle\frac{\partial }{\partial s},$ where $s$ is the global coordinate on $\R$. 

In the proof of such a result, we will use the following fact.
\begin{lemma}\label{s1}
Let $Z$ be  a complete vector field on a manifold $M$.  Then, the following sentences are equivalent: 
\begin{enumerate}
\item 
There is a function $\sigma\in C^\infty(M)$ such  that  $Z(\sigma)=r,$ with $r\in \R-\{0\}.$ 
\item There exist a submanifold $D$ of $M$ of codimension $1$  and a diffeomorphism 
$\varphi:M\to \R\times D$ such that, if 
$t$ is  the global coordinate on $\R,$ 
\begin{equation}\label{T}
T_y\varphi(Z_y)=\displaystyle\frac{\partial}{\partial t}_{|\varphi(y)},
\end{equation} for all $y\in D$.
\end{enumerate}
 In fact, if \mbox{\rm (i)} holds then $D=\sigma^{-1}(0)$ and the inverse map to $\varphi$ is just the restriction to $\R\times D$ of the flow $\Phi^Z$ of the vector field $Z.$  
\end{lemma}

\begin{proof}
Suppose that there is a function $\sigma\in C^\infty(M)$ such  that  $Z(\sigma)=r\not=0.$ Then $0$ is a regular value 
of $\sigma$ and $D=\sigma^{-1}(0)$ is a  submanifold of $M$ of codimension $1$. Now, we consider the smooth map 
$$\varphi: M\to \R\times D,\;\;\; \varphi(y)=\left(\frac{1}{r}\sigma(y), \Phi^Z(-\frac{1}{r}\sigma(y),y)\right),\;\;\; y\in M,$$
where $\Phi^Z:\R\times M\to M$ is the flow of $Z.$ 

Note that, since $Z(\sigma)=r$ then, by integration, we deduce that
\begin{equation}\label{SP}
\sigma(\Phi^Z_t(y))=r t + \sigma(y),\;\;\;  \mbox{ for all } y\in M.
\end{equation} 
As a consequence, $\sigma(\Phi^Z(-\displaystyle\frac{1}{r}\sigma(y),y))=-\sigma(y) + \sigma(y)=0,$ that is, $\Phi^Z(-\displaystyle\frac{1}{r}\sigma(y),y)\in D.$

On the other hand, if $(t,x)\in \R\times D$ then, using (\ref{SP}) and {the fact that} $\sigma(x)=0$, we have that

\begin{eqnarray*}
\varphi(\Phi^Z(t,x))&=&\left(\frac{1}{r}\sigma(\Phi^Z(t,x)), \Phi^Z\left(-\frac{1}{r}\sigma(\Phi^Z(t,x)),\Phi^Z(t,x)\right)\right)\\
&=&\left(t+\frac{1}{r}\sigma(x), \Phi^Z \left(-t - \frac{1}{r}\sigma(x), \Phi^Z(t,x)\right)\right)=(t, x).
\end{eqnarray*}
Moreover,  if $y\in M,$ 
$$\Phi^Z(\varphi(y))=\Phi^Z\left(\frac{1}{r}\sigma(y),\Phi^Z(-\frac{1}{r}\sigma(y),y)\right)=\Phi^Z(0,y)=y.$$
Therefore, $\varphi$ is bijective and $\varphi^{-1}$ is just the restriction of $\Phi^Z$ to $\R\times D.$ Finally, using this last fact and that
 $(T_{(t,x)}\Phi^Z)(\frac{\partial}{\partial t}_{|(t,x)})=Z_{|\Phi^Z(t,x)}$ for $(t,x)\in \R\times D,$ we conclude that 
 $$T_y\varphi(Z_y)=\displaystyle\frac{\partial}{\partial t}_{|\varphi(y)}.$$

Conversely, if there exists a diffeomorphism  $\varphi:M\to D\times \R$ with  $D$ a submanifold  of $M$ 
of codimension $1$ such that (\ref{T}) holds, then the function $\sigma:M\to \R$ is just $\sigma=p_r\circ \varphi,$ 
$p_r:\R\times D \to \R$ being the map $p_r(t,x)=r t.$ The relation $Z(\sigma)=r$ is deduced from (\ref{T}).

\end{proof}

Now, we apply the previous lemma to a particular class of contact Hamiltonian systems.

\begin{proposition}\label{p1}
Let $(C,\eta)$ be a contact manifold  with complete Reeb vector field $\xi$ and  $H:C\to \R$ a Hamiltonian function such that 
$\xi(H)(x)=\gamma\not=0$ for all $x\in C$.
Then,  there exists a  contact isomorphism  $\varphi_1: C \to \R\times S$ 
from the contact manifold  $(C,\eta)$ to the contactification  $( \R\times S, dz + i_S^*\eta)$ of  
the  exact symplectic manifold  $(S,di_S^*\eta)$, i.e.~$\varphi_1$ is a diffeomorphism and 
 $$\varphi_1^*(dz + i_S^*\eta)=\eta,$$
 where $i_{S}:S\to C$ is the inclusion map. Therefore, 
$$T_y\varphi_1(\xi_y)=\displaystyle\frac{\partial }{\partial z}_{|\varphi_1(y)},\quad \mbox{ for all $y\in C.$} $$  
 \end{proposition}
\begin{proof}
In order to obtain the diffeomorphism $\varphi_1: C\to \R \times S$,  we apply Lemma \ref{s1} to the manifold $C$ and the vector field $\xi$. 
Then, there is a  diffeomorphism $\varphi_1: C\to \R \times S$ such that 
$$T_y\varphi_1(\xi_y)=\displaystyle\frac{\partial }{\partial z}_{|\varphi_1(y)}.$$
Moreover, if $\Phi^\xi_{|\R\times S}:\R\times S\to C$  is the restriction of the flow  of $\xi$ to $\R\times S,$ then $\varphi_1$ is just $(\Phi^\xi_{|\R\times S}) ^{-1}$ and 
\begin{equation}\label{xi-t}
(\Phi_{|\R\times S}^\xi)^*(\eta)\left(\frac{\partial }{\partial z}\right)=\eta(\xi)=1.
\end{equation} 
In addition, using that ${\mathcal L}_\xi\eta=0,$ we deduce that 
\begin{equation}\label{L-t}
{\mathcal L}_{\frac{\partial}{\partial z}}\left[(\Phi^\xi_{|\R\times S})^*(\eta)\right]=0.
\end{equation}
Therefore, from (\ref{xi-t}) and (\ref{L-t}), we obtain
\begin{equation}\label{des}
(\Phi_{|\R\times S}^\xi)^*(\eta)=dz+ \alpha
\end{equation}
with $\alpha$ a $1$-form on $S$.

 Now, we consider the map $ i_0:S\to \R\times S,$ given by  $i_0(x)=(0,x).$ Then, it is clear that  the following  diagram
$$
 \xymatrix{S\ar[dr]_{i_S}\ar[rr]^{i_0} &&
\R\times S\ar[dl]^{\Phi_{|\R\times S}^\xi}\\
&C&}
$$
is commutative. Thus, 
from  (\ref{des}), we deduce  that 
$$i_S^*\eta=i_0^*((\Phi_{|\R\times S}^\xi)^*(\eta))=i_0^*(dz) + i_0^*(\alpha)=\alpha. $$
 As a consequence, 
$$(\Phi_{|\R\times S}^\xi)^*(\eta)=dz +i_S^*\eta.$$ Since $\varphi_1=(\Phi^\xi_{|\R\times S})^{-1}$, we conclude that 
$$\varphi_1^*(dz + i_S^*\eta)=\eta.$$
\end{proof} 

If, in addition, we suppose that  {the reparametrization $\displaystyle\frac{X_{H|S}}{\xi(H)\circ i_S}$ of $(X_H)_{|S}$} is complete,  we also have the following result 
 \begin{proposition}\label{p2}
Let $(C,\eta)$ be a contact manifold of dimension $2n+1$ and  $H:C\to \R$ a Hamiltonian function such that
\begin{enumerate}
\item The Reeb vector field $\xi$ is complete.
\item $\xi(H)(x)
\not=0$ for all $x\in S=H^{-1}(0)$. 
\item The reparametrization  $Z=\displaystyle\frac{X_{H|S}}{\xi(H)\circ i_S}$ of the restriction of the Hamiltonian vector field $X_H$ to the submanifold $S$ of $C$ is complete, where $i_{S}:S\to C$ is the inclusion map.
\end{enumerate}
Then, the following sentences are equivalent: 
\begin{enumerate}
\item[(a)] There is a function  $\sigma:S\to \R$ 
such that 
$X_{H|S}(\sigma) = n(\xi(H) \circ i_S).$  
\item[(b)]
There exist a submanifold $B$  of $S$ of codimension $1$ with canonical inclusion $i_B:B \to C$ 
such that $\eta_B=i_B^*\eta$ is a contact structure on $B$ and  a symplectomorphism   $$\varphi_2: S\to \R\times B$$ 
from the symplectic manifold $(S,di_S^*\eta)$ to the symplectification of $(B,\eta_B)$, i.e.
$$\varphi_2^*(\exp(-s)(d\eta_B-ds\wedge \eta_B))=di_S^*\eta,$$
where $s$ is the global coordinate on $\R.$ Moreover, 
\begin{equation}\label{Tvarphi2}
T_y\varphi_2\left(\left( \displaystyle\frac{X_{H|S}}{\xi(H)\circ i_S}\right)_y\right)=\displaystyle\frac{\partial }{\partial s}_{|\varphi_2(y)},
\end{equation}
for all  $y\in S.$
  \end{enumerate} 
 \end{proposition}
\begin{proof} If (b) holds, then  we can consider the function $\sigma:S\to \R$ given by 
$$\sigma=n(pr_1\circ \varphi_2),$$ where $pr_1:\R\times B\to \R$ is the canonical projection on the first factor. Thus, using (\ref{Tvarphi2}), it follows that 
$$Z(\sigma)=\frac{X_{H|S}}{\xi(H) \circ i_S}(\sigma)=n.$$
Next, assume that (a) holds. Then we may apply Lemma \ref{s1}, with  $M=S$, $Z=\displaystyle\frac{X_{H|S}}{\xi(H) \circ i_S}\in {\mathfrak X}(S)$ and $r=n\in \R-\{0\}$ and we deduce that
there exist  a submanifold $B=\sigma^{-1}(0)$ of codimension $1$ of $S$ and a diffeomorphim $\varphi_2: S\to \R\times B$ such that 
\begin{equation}\label{T1}
T_y\varphi_2\left(Z_y\right)=\displaystyle\frac{\partial }{\partial s}_{|\varphi_2(y)}, \quad\mbox{ for $y\in S.$}
\end{equation}
Moreover, from Theorem \ref{casoS}, we conclude that $di_S^*\eta$ is an exact symplectic structure on $S$ and 
\begin{equation}\label{X-D}
X_{H|S}=-(\xi(H)\circ i_S) \Delta,
\end{equation}
where $\Delta$ is the Liouville vector field of the exact symplectic manifold $(S,di_{S}^*\eta)$. This implies that $Z=\frac{X_{H|S}}{\xi(H)\circ i_S}=-\Delta$  and that 
\begin{equation}\label{LX}
{\mathcal L}_{Z}(i^*_S\eta)=-i_S^*\eta.
\end{equation}

Now, if $i_B:B\to C$ is the inclusion map, we will prove that 
\begin{equation}\label{sim}
{\varphi}_2^*(\exp({-s})(di_B^*\eta-ds\wedge i_B^*\eta))=di_S^*\eta.
\end{equation}

In fact, using (\ref{T1}) and (\ref{LX}) and the fact that 
$\Phi^{Z}_{|\R\times B}:\R\times B\to S$
is the inverse map of ${\varphi}_2,$ it follows that 
$${\mathcal L}_\frac{\partial }{\partial s}(\Phi_{|\R\times B}^{Z})^* (i_S^*\eta)=-(\Phi_{|\R\times B}^{Z})^*(i_S^*\eta),$$
or equivalently, 
$${\mathcal L}_{\frac{\partial}{\partial s}} (\exp({s})(\Phi_{|\R\times B}^{Z})^*(i^*_S\eta))=0.$$
This implies that 
\begin{equation}\label{descomp}
\exp({s})(\Phi_{|\R\times B}^{Z})^*(i_S^*\eta)=fds + \alpha,
\end{equation}
where $f\in C^\infty(B)$ and $\alpha\in \Omega^{1}(B).$ Then,  we consider the map $ i_0:B\to \R\times B$, given by  $i_0(x)=(0,x),$ and the following diagram 
$$
 \xymatrix{B\ar[d]_{i_B}\ar[r]^{i_0} &
\R\times B\ar[d]^{\Phi_{|\R\times B}^{Z}}\\
C&\ar[l]_{{i}_S}S}
$$
is commutative. Thus, using 
(\ref{descomp}), we have that 
\begin{equation}\label{alpha}
i_B^*\eta=(i_S\circ \Phi_{|\R\times B}^{Z}\circ i_0)^*\eta=i_0^* ((\Phi_{|\R\times B}^{Z})^*( i_S^*\eta))=i_0^*(\exp({-s})(fds +\alpha))=\alpha.
\end{equation}
On the other hand, again from (\ref{descomp}),
\begin{equation}\label{d}
d(\Phi_{|\R\times B}^{Z})^*(i_S^*\eta)=\exp({-s})(df\wedge ds+d\alpha)-\exp({-s})(ds\wedge \alpha).
\end{equation}
Therefore, 
\begin{equation}\label{i-s}
	\iota_{\frac{\partial}{\partial s}}d(\Phi_{|\R\times B}^{Z})^*(i_S^*\eta)=-\exp({-s})df- \exp({-s})\alpha.
\end{equation}
Moreover, from (\ref{T1}), (\ref{LX}) and since ${\varphi}_2$ is the inverse map to $\Phi_{|\R\times B}^{Z},$ we have that 
\begin{equation}\label{sdP}
\begin{array}{rcl}
\iota_{\frac{\partial}{\partial s}}d(\Phi_{|\R\times B}^{Z})^*(i_S^*\eta)&=&
\iota_{\frac{\partial}{\partial s}}(\Phi_{|\R\times B}^{Z})^*(di_S^*\eta)
=(\Phi_{|\R\times B}^{X_{H|S}})^*(\iota_{Z}di_S^*\eta)\\[8pt]
&=&-(\Phi_{|\R\times B}^{Z})^*(i_S^*\eta)-(\Phi_{|\R\times B}^{Z})^*(d\iota_{Z}i_S^*\eta)\\[8pt]
&=&-(\Phi_{|\R\times B}^{Z})^*(i_S^*\eta),
\end{array}
\end{equation}

where the last equation follows using that

$$\iota_{Z}(i_S^*\eta)=-(\iota_\Delta(i_S^*\eta))=0.$$

So, from (\ref{alpha}), (\ref{i-s}) and (\ref{sdP}), we deduce that 
$$i_B^*\eta=i_0^*(\Phi_{|\R\times B}^{Z})^*(i_S^*\eta)=i_0^*(\exp({-s})df +\exp({-s})\alpha)=
df+i_B^*\eta$$
that is, 
\begin{equation}\label{f}
df=0
\end{equation}
Substituting (\ref{alpha}) and (\ref{f}) in (\ref{d}), we have that 
$$(\Phi_{|\R\times B}^{Z})^*(di_S^*\eta)=\exp({-s})(di_B^*\eta- ds\wedge i_B^*\eta),$$
which implies (\ref{sim}). Thus, since $(\Phi_{|\R\times B}^{Z})^*(di_S^*\eta)$ is a symplectic structure on $\R\times B$, 
we have  that $\eta_B=i_B^*\eta$ is a contact structure on $B$ (see Example \ref{Ex-2.1}).

\end{proof}
\begin{remark}\label{Weinstein}
{\rm Under the same hypotheses (i), (ii) and (iii) as in Proposition \ref{p2}, 
let $\Delta$ be the Liouville vector field of the exact symplectic manifold $(S, d(i_S^*\eta))$. Using (iii) in such a proposition and (\ref{X-H-S-Delta}), we have that $\Delta$ is complete. Now, suppose that (a) in Proposition \ref{p2} holds. This, using again (\ref{X-H-S-Delta}), implies that $\Delta(\sigma) = -n$. So, proceeding as in Remark \ref{Eliashberg-Gromov-Sackel}, we deduce that $(\Delta, -\sigma)$ is a gradient-like pair.  In addition, $(S, d(i_S^*\eta), \Delta, -\sigma)$ is a Weinstein manifold without singular points (see \cite{CiEl,ElGr,Sa,We} for the theory of Weinstein manifolds). 
}
\end{remark}
Propositions \ref{p1} and \ref{p2}  motivate the following definition 
\begin{definition} \label{def:symplecticsandwich}
A {\bf symplectic sandwich with contact bread} consists of 
\begin{enumerate}
\item 
 A  contact manifold  $(C,\eta)$ of dimension $2n+1$ and two submanifolds 
 $S$ and $B$ of codimension $1$ and $2,$ respectively, such that the $2$-form  $d(i_S^*\eta)$ is a 
 symplectic structure on $S$ and the $1$-form $i_B^*\eta$ is a contact structure on $B$, where $i_S:S\to  C$ and $i_B:B\to C$ are the canonical inclusions. 
 \item A contact isomorphism 
 $$\varphi_1:C\to \R \times S$$
 from the contact manifold $(C,\eta)$ to the contactification $(\R\times S,dz+i_S^*\eta)$ of 
 the exact symplectic manifold $(S,di_S^*\eta).$ Here, $z$ is the standard coordinate on $\R$ in $\R\times S.$ 
 \item A symplectic isomorphism
 $$\varphi_2: S \to \R\times B$$
 from the symplectic manifold $(S,di_S^*\eta)$ to the symplectification $(\R\times B, \exp(-s)(di_B^*\eta - ds\wedge i_{B}^*\eta))$ 
 of the contact manifold $(B,i_B^*\eta)$ . Here, $s$ is the standard coordinate on $\R$ in $\R\times B.$ 
 \end{enumerate}
 \end{definition}
 
 \begin{remark}{\rm
  \begin{enumerate}
 \item The map 
 $$({\rm Id}_{\R}\times \varphi_2):C\to \R\times (\R\times B)$$
 is a contact isomorphism from the original contact manifold $(C,\eta)$ to the contactification of the 
 symplectification of the contact manifold $(B,i_B^*\eta)$. Note that the contact structure on $\R\times (\R\times B)$ is $dz+\exp(-s)i_B^*\eta.$ 
 
 \item We have the chain of embeddings (the inclusions $i_S$ and $\widetilde{i}_B$) and projections ($pr_2\circ \varphi_1$ and $pr_2\circ \varphi_2$)
$$
\xymatrix{(C,\eta)\ar[rr]^{pr_2\circ \varphi_1}  &&
(S,d(i_S^*\eta))\ar@<+1ex>[ll]^{i_S}\ar[rr]^{pr_2\circ \varphi_2}   && (B,i_B^*\eta) \ar@<+1ex>[ll]^{\tilde{i}_B}\ar@/_2pc/[llll]^{i_B} 
}
$$

\noindent such that the manifold in the middle of the chain is a symplectic manifold and the manifolds on the left and the one on the right hand of the chain  are contact. 
This is a good motivation for using the terminology {\it a symplectic sandwich with contact bread.} \\
\end{enumerate}

For the symplectic sandwich with contact bread we will use the following notation
$$
\xymatrix{(C,\eta)\ar[rr]^{\kern-20pt \varphi_1}   &&
(\R\times S,dz+i_S^*\eta)\ar[rr]^{\kern-20pt {\rm Id}_\R\times\varphi_2}   &&(\R\times \R\times B,dz+ \exp(-s)i_B^*\eta)}
$$
}
\end{remark}

Now, from Propositions~\ref{p1} and~\ref{p2} and Theorem~\ref{inv-measure-S}, we deduce the main result of this section. 

\begin{theorem}\label{th:main}
Let $(C, \eta, H)$ be a contact Hamiltonian system of dimension $2n+1$ with complete Reeb vector field $\xi$.  
Suppose that $\xi(H)(x)=\gamma\not=0$ for all $x\in C$, 
and that  the restriction $X_{H|S}$ to $S$ of 
the contact Hamiltonian dynamics $X_H$ is complete. If $i_S:S\to C$ is the canonical inclusion and $\sigma:S\to \R$ is a $C^\infty$-function 
on $S$, then $X_{H|S}$ admits an invariant measure 
$$
\nu=\frac{\exp(\sigma)}{\gamma}(di_{S}^*\eta)^n
$$
if and only if we have a symplectic sandwich with contact bread 
$$
\xymatrix{(C,\eta)\ar[rr]^{\kern-20pt \varphi_1}   &&
(\R\times S,dz+i_S^*\eta)\ar[rr]^{\kern-20pt Id_\R\times\varphi_2}   &&(\R\times \R\times B,dz+ \exp(-s)i_B^*\eta)}
$$ and 
$$
T_y\varphi_2\left((X_{H|S})_y\right)={\gamma}\displaystyle\frac{\partial }{\partial s}_{|\varphi_2(y)},
$$
for all  $y\in S.$

\end{theorem}

\begin{example}[\bf Dissipative mechanical systems (re-revisited)]
{\rm In the particular case of Example~\ref{ExDiss2}, in which $V({\bf q})=q^{1}+q^{2}$, the symplectic sandwich with contact bread  is given by 
$$
\xymatrix{(\R\times T^*\R^{2},dz -p_idq^i)\ar[rr]^{ \varphi_1} &&
\R\times (T^*\R^{2},dq^i\wedge dp_i)} 
$$
where $$\varphi_1(z,{\bf q},{\bf p})=\left(\frac{1}{2\gamma} (p_1^2 +p_2^2)+ \frac{1}{\gamma}{(q^{1}+q^{2})} + z, q^1,q^2,p_1,p_2\right).$$
Since $\sigma$ is given by~\eqref{sigma}, the submanifold $B$ is just 
$$B=\left\{\left(-q^2-\frac{1}{\gamma}(p_1+p_{2}),q^2,p_1,p_2\right)\in T^*\R^{2}\right\}\cong 
 \R^{3}$$ 
 and its contact structure is 
 $$\eta_B=(p_1-p_2)dq^2 + \frac{1}{\gamma^2}dp_1 + \left(\frac{p_1-p_2}{\gamma} + \frac{1}{\gamma^2}\right)dp_2\,.$$
From (\ref{XHSex}), we have that the flow $\Phi^{X_{H|S}}: \R\times T^*Q\to T^*Q$ of $X_{H|S}$ is just 
$$
\Phi^{X_{H|S}}(s,({\bf q},{\bf p}))=\left(\frac{1}{\gamma}\left({\bf p}+ \frac{1}{\gamma}\right)(1-\exp({-\gamma s}))-\frac{1}{\gamma}s+{\bf q},\left({\bf p}+ \frac{1}{\gamma}\right)\exp({-\gamma s})-\frac{1}{\gamma}\right)
$$

In this case, a straightforward computation proves that the symplectomorphism 
$$
\xymatrix{(T^*\R^{2},dq^i\wedge dp_i)\ar[rr]^{\hspace{-30pt}\varphi_2} &&
(\R \times B, \exp(-s)(d\eta_B - ds\wedge \eta_B))}
$$
is given by 

\begin{align}
\varphi_2({\bf q},{\bf p})&=\left(\gamma F({\bf p},{\bf q}),\Phi^{X_{H|S}}\big(-F({\bf p},{\bf q}),{\bf q},{\bf p}\big)\right)\notag
\\
&=\left(\gamma F({\bf p},{\bf q}),\frac{1}{2\gamma^{2}}\left[(-2\gamma p_{2}-2)\exp(\gamma F)+2+\gamma^{2}(q^{2}-q^{1})+\gamma(p_{2}-p_{1})\right],\right.\notag\\
&\phantom{=}\qquad\left.\frac{1}{\gamma}\left[\exp(\gamma F)(\gamma p_{1}+1)-1\right],\frac{1}{\gamma}\left[\exp(\gamma F)(\gamma p_{2}+1)-1\right]\right)\,, 
\end{align}
 where $F=\frac{\sigma({\bf q},{\bf p})}{2\gamma}=-\frac{1}{2}(p_1+p_2 +\gamma(q^1 + q^2)).$
 }
 
 \end{example}
 
 \section{Conclusions and future work}\label{conclusions-future-work}

In this work we proved various results concerning contact Hamiltonian systems and the existence of invariant measures for such dynamics.
In particular, we proved that in the open subset $U$ where the Hamiltonian function $H$ has no zeros,  
the Hamiltonian vector field $X_H$ is the Reeb vector field of a conformal change of the original contact form $\eta$ (Theorem~\ref{casoU}). 
This directly implies the existence of an invariant measure for $X_{H|U}$,  recovering the corresponding result in \cite{BrTa} (see Corollary~\ref{BrTa}). 
On the other hand, assuming that the Reeb vector field $\xi$ associated with $\eta$ is transverse to the zero level set $S$ of $H$, 
we showed that $S$ is endowed with an exact symplectic structure and 
that $X_{H|S}$ is, up to reparametrization, the Liouville vector field of the exact symplectic manifold $S$ (Theorem~\ref{casoS}). 
Then we used this result to characterize the general condition for the existence of an invariant measure for $X_{H|S}$ (see Theorem~\ref{inv-measure-S}).
Finally, under the assumptions that $\xi(H)=\gamma\in \mathbb R-\{0\}$ and that both $\xi$ and $X_{H|S}$ are complete, 
we proved in Theorem~\ref{th:main}, 
as the main result of this paper, that $X_{H|S}$ admits an invariant measure if and only if the original contact manifold $C$ is a 
symplectic sandwich with contact bread (see Definition~\ref{def:symplecticsandwich})
and $X_{H|S}$ admits a suitable global rectification in the sandwich. 

Our contributions are closely related to the theory of Weinstein and convex contact manifolds without singular points 
(see Remarks \ref{Eliashberg-Gromov-Sackel} and \ref{Weinstein}). Therefore, taking into account that a lot of relevant results have been obtained for the previous manifolds 
in the presence of singular points (see \cite{CiEl,Sa}), it would be interesting to use such results to describe the dynamical behavior 
of the vector field $X_{H|S}$ in the presence of a Weinstein structure on the exact symplectic manifold~$S$. 
We will try to describe the dynamics in the region where there are no singular points and in a neighborhood of each singular point.  On the other hand, 
as another future work, we expect to generalize some of the results to the case where $\xi(H)\neq 0$
is not necessarily a constant (or, more generally, to the case when $(\xi, H)$ is a gradient-like pair following the ideas in \cite{Sa}) and to understand the relationship 
with invariant measures stemming from the Lagrangian description of these systems~\cite{DeLeon2019,Gaset2,Gaset1,Val1,wang2016,wang2019}. 
Moreover, given the relevance of contact Hamiltonian systems in many areas of science, we believe that our work will have 
interesting consequences 
in various contexts, especially in statistical physics and thermodynamics~\cite{bravetti2019contactthermo,BrTaPRE,Grmela1997,Morrison1986,Simo20202}, statistics~\cite{Mike}, 
and inflationary cosmology~\cite{DSloan1,DSloan2}.

\appendix
\section{\;\;}

 In Lemma \ref{s1}, we have found the  necessary and sufficient conditions for the existence of a function $\sigma\in C^\infty(M)$ 
 such that $Z(\sigma)=r\not=0$, when the vector field $Z$ is complete. In this appendix we will show what happens when  we do not assume  this completeness condition. 

\begin{proposition}
Let $Z$ be  a vector field on a manifold $M$ and $U$ the maximal open subset of $\R\times M$ such that the flow of  $Z $ is defined on $U.$  
Then, the following sentences are equivalent:
\begin{enumerate}
\item 
There is a function $\sigma\in C^\infty(M)$ such that the $\mbox{graph}$ of the function $-\displaystyle\frac{\sigma}{r}$ is included into $U$ and $Z(\sigma)=r.$
\item There exist a submanifold $D$ of $M$ of codimension $1$, an open set $V$ of ${\R}\times D$
 and a diffeomorphism 
$$\varphi:M\to V\subset {\R}\times D$$ such that, 
\begin{enumerate}
\item $(-pr_1(\varphi(y)),y)\in U$ for all $y\in M, $ where $pr_1:\R\times D\to \R$ is the canonical projection on the first factor. 
\item $V$ is a submanifold of $U$ of codimension $1.$
\item $T_y\varphi(Z_y)=\frac{\partial}{\partial t}_{|\varphi(y)}$ for all $y\in M,$ with $t$ the standard global coordinate on $\R$. 
\end{enumerate}
\end{enumerate}
\end{proposition}
\begin{proof} 
As in the proof of Lemma \ref{s1}, if we suppose (i), then we deduce that $D=\sigma^{-1}(0)$ is the submanifold of $M$ 
which we are looking for. Now, we consider the submanifold  of $U$ of codimension $1$
$$V=\{(t,y)\in U/\sigma(y)=0\}=(\sigma\circ pr_2)^{-1}(0)$$
where $pr_2:U\subseteq \R\times M\to M$ is the restriction to $U$ of the canonical projection $\R\times M\to M$ on the second factor. 
Since $V=U\cap ({\R}\times D),$ then $V$ is an open subset of ${\R}\times D$. The diffeomorphism $\varphi$  is constructed as in Lemma \ref{s1} 
$$\varphi: M\to V\subseteq {\R}\times D,\;\;\;\varphi(y)=\left(\frac{1}{r}\sigma(y),  \Phi^Z\left(-\frac{1}{r}\sigma(y),y\right)\right),$$
where $\Phi^Z:U\subseteq \R \times M \to M$ is the flow of $Z$. Note that, using that the graph of 
$-\displaystyle\frac{\sigma}{r}$ is contained in $U$
we have that $\left(-\displaystyle\frac{\sigma(y)}{r},y\right)\in U$ for all $y\in U$. 
This implies that 
$$\left(\displaystyle\frac{1}{r}\sigma(y),  \Phi^Z\left(-\displaystyle\frac{1}{r}\sigma(y),y\right)\right)\in~U.$$ 
Moreover, as in  the proof of Lemma \ref{s1},
$$\sigma\left(\Phi^Z\left(-\frac{1}{r}\sigma(y),y\right)\right)=0,$$
 (that is, $\varphi(y)\in \R\times D$) and the restriction of $\Phi^Z$ to $V$ is the inverse map of $\varphi.$ 
 
Conversely, suppose that $D$ is  a submanifold of $M$ of codimension $1$, $V$ is an open 
subset of $\R\times D$ and $\varphi: M\to V\subseteq \R\times D$ is a diffeomorphim such that  (a), (b) and (c) hold. 
Then, we  consider the function $\sigma=p_r\circ \varphi$, where $p_r:\R\times D\to \R$ is the function $p_r(t,x)=rt.$ 
From (a) we have that $graph(-\frac{1}{r}\sigma)\subseteq U$ and, using (b) and (c), we conclude that $Z(\sigma)=r.$

\end{proof} 

Using this result  we deduce the corresponding versions of Propositions \ref{p1} and \ref{p2}, when the completeness condition does not hold.  

\begin{proposition} 
Let $(C,\eta, H)$ be a contact Hamiltonian system  and $U_1$ the maximal open subset of $\R\times C$ such that the flow of  
the Reeb vector field $\xi $ is defined on $U_1.$ We suppose that  $\xi(H)(x)=\gamma\not=0$ for all $x\in C$ 
and 
that the graph of the function $-\frac{H}{\gamma}$ is contained in $U_1.$  Then,  there exist an open set $V_1$ of $\R\times S$ and  a  
contact isomorphism  $\varphi_1: C \to V_1\subseteq \R\times S$ from the contact manifold  $(C,\eta)$ to the contact manifold  $( V_1, i_{V_1}^*(dz + i_S^*\eta)),$  i.e.
 $$\varphi_1^*(i_{V_1}^*(dz + i_S^*\eta))=\eta,$$
 where $i_{S}:S\to C$ and $i_{V_1}:V_1\to \R\times S$ are the corresponding inclusion maps. Therefore, 
$$T_y\varphi_1(\xi_y)=\displaystyle\frac{\partial }{\partial z}_{|\varphi_1(y)}, \mbox{ for all $y\in C.$} $$  
 \end{proposition}

 \begin{proposition}
Let $(C,\eta,H)$ be a contact Hamiltonian system of dimension $2n+1$ and  $U_1$ the maximal open subset of $\R\times C$ 
such that the flow of  the Reeb vector field $\xi $ is defined on $U_1.$ Suppose that $\xi(H)(x)\not=0$ for all $x\in S=H^{-1}(0).$ 
Then, the following sentences are equivalent: 
\begin{enumerate}
\item[(a)] There is a function  $\sigma:S\to \R$ on $S$ such that $X_{H|S}(\sigma) = n(\xi(H) \circ i_S)$  
and $$graph\left(-\frac{\sigma}{n}\right)\subseteq U_2,$$  with $U_2$ 
the maximal open subset of $\R\times S$ such that the flow of $\frac{X_{H|S}}{\xi(H)\circ i_S}$ is defined on $U_2.$ 
\item[(b)]
There exist a submanifold $B$  of $S$ of codimension $1$, with canonical inclusion $i_B:B\to  C$, such that $\eta_B=i_B^*\eta$ is a contact structure on $B$,   
an open subset $V_2$ of $\R\times B$  and  a symplectic isomorphism   $$\varphi_2: S\to V_2\subseteq \R\times B$$ 
from the symplectic manifold $(S,di_S^*\eta)$ to  $\left(V_2,i_{V_2}^*(\exp(-s)(d\eta_B-ds\wedge \eta_B))\right)$, i.e.
$$\varphi_2^*(i_{V_2}^*(\exp(-s)(d\eta_B-ds\wedge \eta_B)))=d(i_S^*\eta),$$
where $s$ is the global coordinate of $\R$ and $i_{V_2}:V_2\to \R\times B$ is the inclusion map. Moreover, $(pr_1(\varphi_2(y)),y)\in U_2$ and 
 $$T_y\varphi_2\left (\left(\frac{X_{H|S}}{\xi(H)\circ i_S}\right)_y\right)=\displaystyle\frac{\partial }{\partial s}_{|\varphi_2(y)},$$
for all  $y\in S.$ 
 \end{enumerate} 
 \end{proposition}

 \end{document}